\documentclass[aps,pre,twocolumn,groupedaddress]{revtex4-2}
\usepackage{tikz}
\usepackage{footnote}

\usepackage{caption}
\usepackage{amsmath,amssymb}
\usepackage{cleveref}

\usepackage[noend]{algpseudocode} 
\usepackage{algorithm}
\usepackage{amsthm}

\newtheorem{theorem}{Theorem}
\newtheorem{example}[theorem]{Example}

\usepackage{bbm}

\crefname{figure}{figure}{figures}
\Crefname{figure}{Figure}{Figures}
\crefname{algorithm}{Algorithm}{figures}
\Crefname{algorithm}{Algorithm}{Figures}
\crefname{corollary}{Corollary}{corollary}
\Crefname{corollary}{Corollary}{corollary}
\crefname{section}{Section}{section}
\Crefname{section}{Section}{section}

\newcounter{comments}

\begin{document}


\title[Uniform sampling of canalizing Boolean functions]{Uniform sampling of canalizing Boolean functions reveals hidden biases in Boolean network analysis}

\author{Ahana Ghosh}
\affiliation{Iowa State University}

\author{Claus Kadelka}
\email{ckadelka@iastate.edu}
\affiliation{Iowa State University}

\date{\today}

\begin{abstract}
Boolean networks are widely used to model gene regulatory systems. Their structural and dynamical properties are commonly interpreted by comparison with ensembles of random Boolean networks generated by sampling Boolean functions for individual nodes. Canalizing and nested canalizing functions, in which one or more regulatory inputs dominate the output, capture an important feature of gene regulation. These functions are typically generated by sampling their defining parameters uniformly at random. Because multiple parameterizations can represent the same Boolean function, however, this procedure induces a biased distribution over functions and consequently over null models. We develop efficient algorithms for uniformly sampling Boolean functions with prescribed canalizing depth, thereby correcting this systematic bias. Using these unbiased null models, we show that the sampling measure substantially alters function- and network-level properties. Whereas parameter-uniform sampling yields nested canalizing functions with expected average sensitivity one, these sensitivities increase with degree under function-uniform sampling and approach 1.183. These differences alter expectations for robustness, attractor structure, and stability. Reanalysis of 122 published Boolean gene regulatory network models reveals substantially stronger enrichment of low-sensitivity canalizing architectures than previously recognized. Widely used parameter-based null models therefore systematically underestimate baseline sensitivity and overestimate the stabilizing role of canalization.
\end{abstract}


\maketitle

\section*{Key Messages}
\begin{itemize}
\item We develop efficient algorithms for uniform sampling of Boolean functions with exact and minimal canalizing depth, correcting biases inherent in widely used parameter-based sampling methods.
\item We show that parameter-uniform sampling systematically overrepresents low-sensitivity, highly canalizing structures, thereby distorting baseline expectations for key function-level metrics such as average sensitivity.
\item We demonstrate that these sampling biases propagate to Boolean network analyses, affecting conclusions about long-term dynamics, stability, and the inferred enrichment of canalizing architectures in gene regulatory networks.
\end{itemize}

\section{Introduction}
Boolean networks are a widely used framework for modeling gene regulatory systems, where nodes represent genes and Boolean functions encode regulatory logic~\cite{KAUFFMANfirstpaper}. A central component of such analyses is the construction of appropriate null models, typically ensembles of random Boolean networks generated from prescribed classes of Boolean functions, against which structural and dynamical properties can be compared. In particular, canalizing and nested canalizing functions (NCFs) have received sustained attention due to their biological relevance and their association with robustness and stability in network dynamics~\cite{he2016stratification,jarrah2007nested,kadelka2026canalization,kadelka2017influence,kauffman1974large,kauffman2003random,li2013boolean}. Random ensembles of canalizing Boolean functions are frequently used to assess whether observed features of gene regulatory networks deviate from randomness~\cite{bavisetty2026attractors,daniels2018criticality, kadelka2024meta, kadelka2024canalization}.

Random ensembles of canalizing Boolean functions are commonly generated by sampling canalizing inputs, canalized outputs, and variable orderings uniformly at random~\cite{jarrah2007nested,kadelka2024meta,li2013boolean}. This strategy is implemented, for example, in the widely used \texttt{BoolNet} package~\cite{mussel2010boolnet}. However, because multiple parameterizations can represent the same Boolean function, parameter-uniform sampling induces a non-uniform distribution over distinct Boolean functions. The underlying reason is combinatorial: the mapping from parameterizations to Boolean functions is many-to-one, and different functions admit vastly different numbers of representations. 

Consequently, parameter-uniform sampling systematically overrepresents functions with large canalizing layers and few total layers, biasing random ensembles toward low-sensitivity, highly canalizing functions.

Although the combinatorial structure of canalizing Boolean functions has been extensively characterized~\cite{dimitrova2022revealing,he2016stratification,kadelka2017multistate,kadelka2017influence,li2013boolean}, these results have not led to scalable methods for sampling Boolean functions uniformly from biologically relevant canalizing function classes. To address this problem, we develop efficient exact sampling algorithms for uniformly generating Boolean functions with prescribed canalization properties, correcting the sampling bias induced by parameter-uniform sampling. We then use these algorithms to quantify how the choice of sampling measure alters the statistical properties of Boolean functions, the dynamics of random Boolean networks, and the interpretation of biological gene regulatory network models. Our results demonstrate that correcting this sampling bias is essential for constructing appropriate null models and accurately assessing the stabilizing role of canalization in gene regulatory networks.

\section{Results}

\subsection{The origin of sampling bias}

Canalizing Boolean functions are commonly classified according to their
\emph{canalizing depth} and \emph{canalizing layer structure}
\cite{he2016stratification,kadelka2017influence,layne2012nested}.
A Boolean function has canalizing depth $k$ if $k$ variables can be
recursively fixed before evaluating a remaining core function, and
nested canalizing functions (NCFs) are precisely those with canalizing
depth equal to the number of inputs \cite{kauffman2003random}. Every
canalizing function admits a unique decomposition into consecutive
canalizing layers, described by its layer structure
$\lambda=(k_1,\ldots,k_r)$, where $k_i$ denotes the number of variables
in the $i$th layer and $k_1+\cdots+k_r=k$.

The source of the sampling bias lies in the combinatorial multiplicity
of these layer structures. As shown in Appendices \ref{DP Recursion} and \ref{Correctness}, for a
fixed number of inputs $n$, the number of distinct canalizing Boolean
functions with layer structure $\lambda$ is proportional to

\[
\omega(\lambda)=
\frac{1}{\prod_{i=1}^r k_i!}.
\]

Thus, functions with many small canalizing layers are substantially
more abundant than functions with few, large layers. Importantly, this
proportionality holds both for functions with exact canalizing depth
$k$ and for functions with canalizing depth at least $k$.

For example, the four possible layer structures of four-input nested
canalizing functions, $(4)$, $(1,3)$, $(2,2)$, and $(1,1,2)$, differ
by more than an order of magnitude in abundance (table~\ref{tab:layer_weights}).

\begin{table}
\centering
\begin{tabular}{c c}
\hline
Layer structure & Relative number of functions \\
\hline
$(4)$       & $4!/4! = 1$ \\
$(1,3)$     & $4!/(1! \cdot 3!) = 4$ \\
$(2,2)$     & $4!/(2! \cdot 2!) = 6$ \\
$(1,1,2)$   & $4!/(1! \cdot 1! \cdot 2!) = 12$ \\
\hline
\end{tabular}
\caption{Relative abundance of $4$-input NCFs with specific layer structure.}
\label{tab:layer_weights}
\end{table}

Current approaches generate canalizing Boolean functions by sampling their defining parameters---canalizing inputs, canalized outputs, variable orderings, and core function---uniformly at random. Because the mapping from parameterizations to Boolean functions is many-to-one, this procedure samples parameterizations uniformly rather than Boolean functions uniformly. In particular, parameter-uniform sampling assigns equal probability to every layer structure, whereas function-uniform sampling must assign probabilities proportional to $\omega(\lambda)$. Consequently, parameter-uniform sampling systematically overrepresents functions with few, large canalizing layers and biases random ensembles toward these low-sensitivity, highly canalizing functions.

For the four-input example, function-uniform sampling requires (see table~\ref{tab:layer_weights})

\[
12\mathbb P((4))
=
3\mathbb P((1,3))
=
2\mathbb P((2,2))
=
\mathbb P((1,1,2)),
\]

yielding

\[
\mathbb P((4))
=
\frac{1}{23},\
\mathbb P((1,3))
=
\frac{4}{23},\
\mathbb P((2,2))
=
\frac{6}{23},\
\mathbb P((1,1,2))
=
\frac{12}{23}.
\]

Thus, the layer structure $(1,1,2)$, which corresponds to the largest number of distinct NCFs, must be sampled most frequently, whereas the layer structure $(4)$, which corresponds to the smallest number of distinct NCFs, must be sampled least frequently. 
This inversion of combinatorial multiplicities is the key to achieving uniform sampling over Boolean functions and forms the basis of the algorithms developed below.

\subsection{Uniform sampling of canalizing Boolean functions}

The preceding analysis shows that function-uniform sampling requires layer structures to be generated with probabilities proportional to their combinatorial weight $\omega(\lambda)$ rather than uniformly. Directly enumerating all layer structures is, however, infeasible for large numbers of variables because their number grows exponentially with the canalizing depth. We therefore develop efficient exact sampling algorithms that achieve function-uniform sampling without explicit enumeration.

\paragraph{Functions with prescribed canalizing depth.}

We first consider Boolean functions with exact canalizing depth $k$, including nested canalizing functions ($k=n$). Since the canalizing layer structure is completely determined by the sequence of canalized output values, uniform function sampling can be achieved by sampling this sequence with carefully chosen probabilities while continuing to sample canalizing input values, variable orderings, and the non-canalizing core function uniformly at random.

To accomplish this, we construct the canalizing layer structure
sequentially using a dynamic programming algorithm. Beginning with the
first canalized output value, the algorithm repeatedly decides whether
the current canalizing layer should be extended or whether a new layer
should be started. These decisions are based on recursively computed
weights that exactly compensate for the combinatorial multiplicities of
all remaining layer structures. Consequently, every canalizing layer
structure is generated with probability proportional to
$\omega(\lambda)$, yielding uniform sampling over distinct Boolean
functions. Details of the dynamic programming recursion and a proof of
correctness are provided in Appendices \ref{DP Recursion} and \ref{Correctness}.

\paragraph{Functions with prescribed minimal canalizing depth.}

For general $k$-canalizing Boolean functions, whose canalizing depth is
at least $k$, an additional correction is required. If the sampled core
function is itself canalizing, its first canalizing layer may merge with
the final sampled layer, causing certain Boolean functions to be
generated more frequently than others. We correct this overcounting by
an efficient rejection-sampling procedure that accepts merged
representations with probabilities inversely proportional to their
multiplicity. This correction restores uniform sampling over all
Boolean functions with prescribed minimal canalizing depth.

The dynamic programming recursion requires only $O(k^2)$ preprocessing, after which individual layer structures can be generated in linear time. The rejection step is also efficient: its rejection probability is at most $5/8$ and decreases rapidly with the number of core variables (Appendix~\ref{random sampling}). Together, these algorithms enable exact function-uniform sampling across biologically relevant classes of canalizing Boolean functions.

\subsection{Function-uniform sampling revises properties of canalizing Boolean functions}

Function-uniform sampling substantially changes the statistical properties of randomly generated canalizing Boolean functions. To quantify these effects, we compared parameter-uniform and function-uniform ensembles of nested canalizing functions (NCFs) with $n\in\{2,\ldots,8\}$ inputs. Unless stated otherwise, all reported averages are taken over uniformly sampled Boolean functions within the corresponding ensemble.

\begin{figure}
    \centering
    \includegraphics[width=0.96\columnwidth]{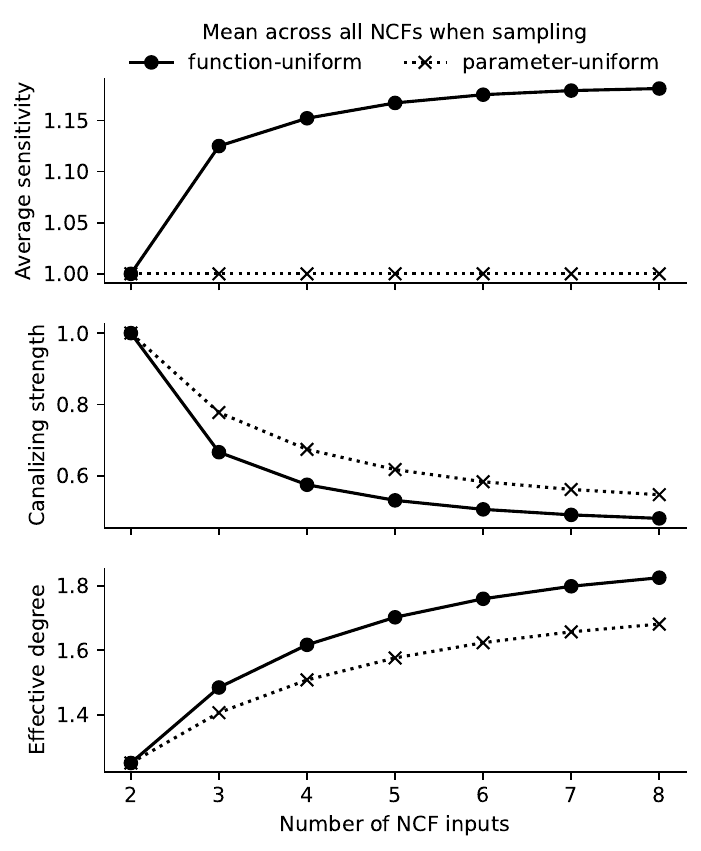}   
    \caption{Function-uniform sampling alters the average properties of nested canalizing functions. Mean average sensitivity (top), mean canalizing strength (middle), and mean effective degree (bottom) are shown as functions of the number of inputs under function-uniform (solid lines) and parameter-uniform (dotted lines) sampling.}
    \label{fig:impact_summary}
\end{figure}

Average sensitivity provides a particularly informative example because it is closely linked to perturbation propagation and dynamical stability in Boolean networks~\cite{shmulevich2004activities,derrida1986random}.  Whereas parameter-uniform sampling yields NCFs with expected average sensitivity one independent of the number of inputs, function-uniform sampling produces sensitivities that increase with degree and approach approximately $1.183$ (figure~\ref{fig:impact_summary}). Thus, the widely cited value of one reflects the sampling measure rather than an intrinsic property of nested canalizing functions.

The choice of sampling measure similarly affects other commonly used measures of canalization. Canalizing strength and effective degree quantify collective canalization and input redundancy; higher canalizing strength and lower effective degree correspond to stronger canalization~\cite{kadelka2023collectively,gates2021effective}.
Compared with parameter-uniform sampling, function-uniform sampling yields lower canalizing strength and higher effective degree (figure~\ref{fig:impact_summary}). Together, these results demonstrate that parameter-uniform sampling systematically favors Boolean functions exhibiting stronger canalization and lower sensitivity than expected under a uniform distribution over distinct functions.

The origin of these differences can be understood by examining the relationship between canalizing layer structure and sampling weight. Because parameter-uniform sampling assigns equal probability to parameterizations rather than functions, the relative sampling weight of a Boolean function decreases approximately exponentially with its average sensitivity (figure~\ref{fig:impact_detailed}). Similar relationships are observed for canalizing strength and effective degree, reflecting the underlying combinatorial dependence of parameter multiplicity on canalizing layer structure. Parameter-uniform sampling therefore suppresses highly sensitive canalizing functions while overrepresenting strongly canalizing, low-sensitivity functions.

\begin{figure*}
    \centering
    \includegraphics[width=0.75\textwidth]{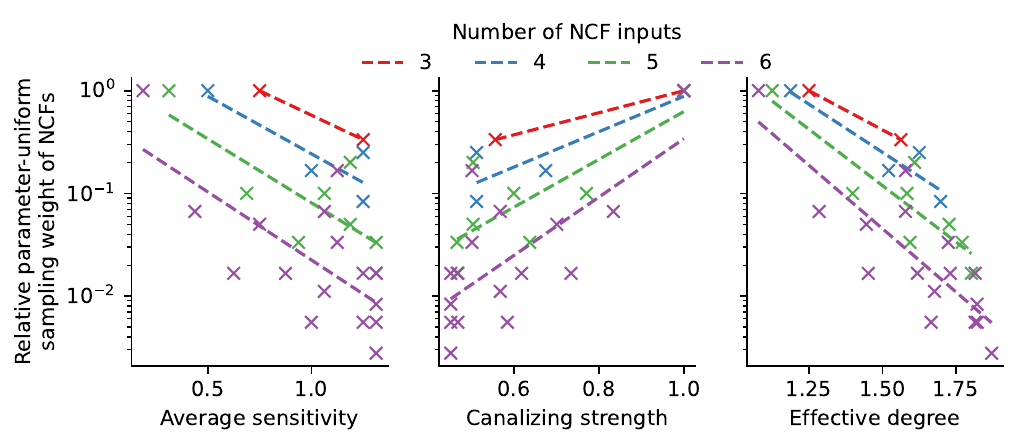}
    \caption{Parameter-uniform sampling suppresses high-sensitivity nested canalizing functions. Relative parameter-uniform sampling weight is shown as a function of average sensitivity (left), canalizing strength (middle), and effective degree (right). Each point represents an NCF layer structure, and dashed lines show exponential fits separately for each number of inputs.}
    \label{fig:impact_detailed}
\end{figure*}

\subsection{Function-uniform sampling alters Boolean network dynamics}

The differences between parameter-uniform and function-uniform sampling naturally propagate from individual Boolean functions to the dynamics of Boolean networks. Because parameter-uniform sampling preferentially selects low-sensitivity, highly canalizing functions, it generates random Boolean networks with systematically different dynamical properties than function-uniform sampling.

\begin{figure}
    \centering
    \includegraphics[width=0.85\columnwidth]{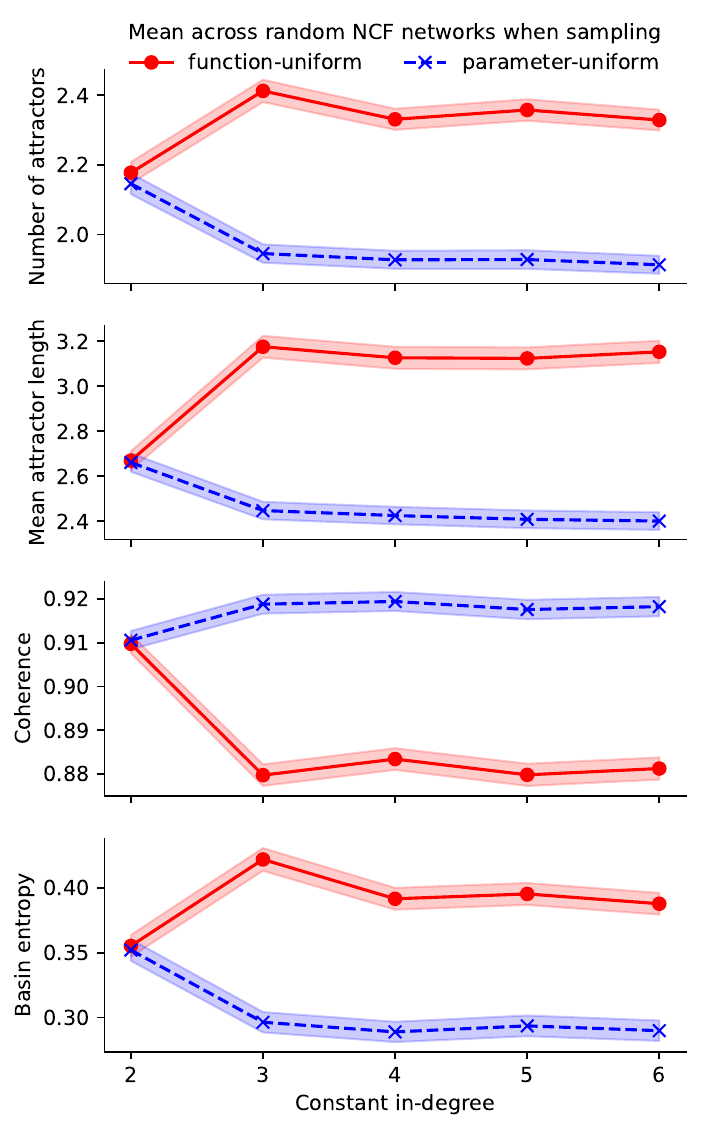}
    \caption{Parameter-uniform sampling favors more ordered Boolean network dynamics. Mean number of attractors, mean attractor length, mean coherence, and mean basin entropy are shown as functions of constant in-degree for networks governed by NCFs under function-uniform (solid red lines) and parameter-uniform (dashed blue lines) sampling. Shaded regions indicate 95\% confidence intervals.}
    \label{fig:dynamics}
\end{figure}

To quantify these effects, we generated ensembles of random Boolean networks composed exclusively of NCFs sampled under either parameter-uniform or function-uniform sampling. For each constant in-degree $n\in\{2,\ldots,6\}$, we generated 10,000 networks consisting of $N=12$ nodes and analyzed their exact synchronous dynamics using BoolForge~\cite{kadelka2026boolforge}.

Compared with function-uniform sampling, parameter-uniform sampling consistently produced networks with fewer attractors, shorter attractor cycles, higher coherence, and lower basin entropy (figure~\ref{fig:dynamics}). Notably, all four dynamical measures shift consistently toward a more ordered regime. This pattern agrees with the established relationship between canalization, reduced sensitivity, and dynamical stability in Boolean networks~\cite{kauffman2004genetic,shmulevich2004activities,karlsson2007order,krawitz2007basin,bavisetty2026attractors}. The coordinated shift across these dynamical properties is consistent with the overrepresentation of low-sensitivity functions under parameter-uniform sampling, which suppresses perturbation propagation and favors more ordered network dynamics.

These results demonstrate that the choice of sampling measure systematically influences commonly studied dynamical properties of random Boolean networks. Consequently, network-level conclusions drawn from parameter-uniform null models cannot be interpreted solely as consequences of canalization, but also reflect the underlying sampling measure used to generate the Boolean functions.

\subsection{Function-uniform null models reveal stronger enrichment of canalizing architectures in biological networks}

The impact of the sampling measure extends beyond random Boolean network ensembles to the interpretation of biological gene regulatory networks. Random canalizing Boolean functions are widely used as null models to assess whether regulatory logic observed in biological systems differs from random expectation. If the null model itself is biased, however, the inferred degree of biological enrichment will also be biased.

To quantify this effect, we revisited the meta-analysis of 122 published Boolean gene regulatory network models performed in~\cite{kadelka2024meta}. For NCFs with $n\in\{3,4,5\}$ inputs, we compared the observed abundance of each canalizing layer structure against its expected abundance under both parameter-uniform and function-uniform sampling.

The choice of sampling measure substantially changes the inferred enrichment of biologically observed canalizing architectures (figure~\ref{fig:specific_layers}). Relative to function-uniform expectations, low-sensitivity layer structures are considerably more enriched, whereas highly layered and more sensitive structures are less abundant than suggested by parameter-uniform null models. This difference becomes increasingly pronounced with increasing in-degree.

\begin{figure}
    \centering
    \includegraphics[width=0.99\columnwidth]{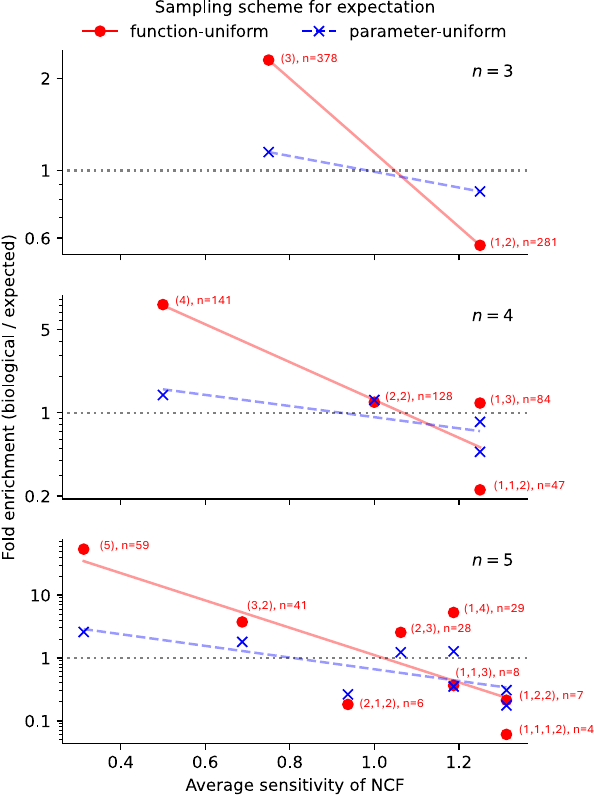}
    \caption{Function-uniform null models reveal stronger enrichment of low-sensitivity canalizing architectures in biological networks. Fold enrichment of NCF layer structures relative to function-uniform (red circles) and parameter-uniform (blue crosses) expectations is shown as a function of average sensitivity for NCFs with three (top), four (middle), and five (bottom) inputs. Red data points are labeled by layer structure and the corresponding number of NCFs observed in biological models. Colored lines show exponential fits to the enrichment trends, and the horizontal dotted line indicates the null expectation of no enrichment.}
    \label{fig:specific_layers}
\end{figure}

These results indicate that gene regulatory networks are enriched not only for canalization, but even more strongly for low-sensitivity canalizing architectures than previously recognized. Consequently, parameter-uniform null models systematically underestimate the extent to which biological regulatory logic deviates from random expectation.

\section{Discussion}

The construction of null models is a fundamental component of Boolean network analysis because conclusions about biological organization are inherently comparative. Observed structural and dynamical properties are typically interpreted relative to ensembles of randomly generated Boolean networks that preserve selected characteristics of the biological system. Our results demonstrate that the probability measure used to generate these ensembles is itself a modeling assumption with substantial consequences. Parameter-uniform and function-uniform sampling do not represent equivalent implementations of the same null model; rather, they correspond to distinct distributions over Boolean functions and therefore to distinct biological hypotheses.

Function-uniform sampling provides a principled baseline whenever the objective is to compare biological regulatory logic with the space of all Boolean functions satisfying prescribed canalization constraints. In contrast, parameter-uniform sampling is not inherently incorrect. Instead, it represents a different generative model in which parameterizations, rather than Boolean functions, are sampled uniformly. Such a model may be appropriate when the parameterization itself reflects the biological process of interest. However, it should not be interpreted as a uniform null model over Boolean functions.

The distinction between these two sampling measures has important practical consequences. We showed that parameter-uniform sampling systematically favors low-sensitivity, highly canalizing functions, leading to random Boolean networks with enhanced dynamical stability. Consequently, previous studies comparing biological networks against parameter-uniform null models have underestimated the extent to which biological regulatory logic is enriched for stabilizing canalizing architectures. Correcting this bias reveals that biological networks depart more strongly from random expectation than previously appreciated.

Although we focused on canalizing Boolean functions, the underlying issue is considerably more general. Whenever multiple parameterizations correspond to the same mathematical object, sampling parameterizations uniformly need not induce a uniform distribution over the objects themselves. Similar sampling biases may therefore arise in other classes of Boolean functions and more broadly in the construction of null models across computational biology and network science. Identifying the appropriate probability measure should thus be regarded as a central component of null model design rather than a technical implementation detail.

The algorithms presented here provide an efficient solution to this problem for biologically relevant classes of canalizing Boolean functions. By enabling exact function-uniform sampling, they establish a principled foundation for future studies of Boolean network dynamics, robustness, and the statistical organization of gene regulatory systems.

\section{Methods}

\subsection{Canalizing Boolean functions}

We briefly summarize the definitions and notation used throughout the paper; more detailed treatments are available in~\cite{he2016stratification,kadelka2017influence,li2013boolean}. A Boolean function $f:\{0,1\}^n\to\{0,1\}$ is \emph{canalizing} if fixing at least one input variable $x_i$ to a particular \emph{canalizing input value} $a\in\{0,1\}$ sets $f$ to a \emph{canalized output value} $b\in\{0,1\}$, irrespective of the remaining inputs~\cite{kauffman1974large}.

More generally, a Boolean function $f$ is \emph{$k$-canalizing} if there exist distinct variables $x_{\sigma(1)},\ldots,x_{\sigma(k)}$, canalizing input values $a_1,\ldots,a_k\in\{0,1\}$, canalized output values $b_1,\ldots,b_k\in\{0,1\}$, and a core function $g:\{0,1\}^{n-k}\to\{0,1\}$, with $g\not\equiv b_k$, such that
\[
f =
\begin{cases}
b_1, & x_{\sigma(1)}=a_1,\\
b_2, & x_{\sigma(1)}\neq a_1,\ x_{\sigma(2)}=a_2,\\
\vdots\\
b_k, & x_{\sigma(1)}\neq a_1,\ldots,x_{\sigma(k-1)}\neq a_{k-1},\ x_{\sigma(k)}=a_k,\\
g, & x_{\sigma(1)}\neq a_1,\ldots,x_{\sigma(k)}\neq a_k.
\end{cases}
\]
The \emph{canalizing depth} of $f$ is the largest $k$ for which such a representation exists~\cite{he2016stratification,layne2012nested}. A \emph{nested canalizing function} (NCF) has canalizing depth $n$; equivalently, all variables are canalizing in succession and the remaining core is the constant function $g=1-b_n$~\cite{kauffman2003random}.

Every canalizing Boolean function admits a unique decomposition of its canalizing variables into consecutive \emph{canalizing layers}~\cite{he2016stratification}. Its \emph{layer structure} is denoted by $\lambda=(k_1,\ldots,k_r)$, where $k_i$ is the number of canalizing variables in the $i$th layer and $k_1+\cdots+k_r=k$, with $k$ the canalizing depth~\cite{kadelka2017influence}. The layer structure is determined by the sequence of canalized output values $b_1,\ldots,b_k$: consecutive variables $x_{\sigma(i)}$ and $x_{\sigma(i+1)}$ belong to the same layer if and only if $b_i=b_{i+1}$. Thus, sampling the sequence of canalized output values is equivalent to sampling the canalizing layer structure, an observation that underlies the function-uniform sampling algorithm developed in this paper.

\subsection{Generation and analysis of Boolean functions}

All canalizing Boolean functions were generated and analyzed using the BoolForge software package~\cite{kadelka2026boolforge}. Under parameter-uniform sampling, we followed the standard construction implemented in \texttt{BoolNet}~\cite{mussel2010boolnet} and elsewhere: canalizing variables (whenever $k<n$), their ordering, canalizing input values, and canalized output values were sampled uniformly at random, together with the core function when applicable. Function-uniform sampling used the algorithms developed in this work.

For the function-level comparison of the two sampling measures, we considered all NCFs with $n\in\{2,\ldots,8\}$ inputs when comparing mean average sensitivity, canalizing strength, and effective degree (figure~\ref{fig:impact_summary}). The average sensitivity of a Boolean function is the expected number of output changes caused by flipping a single input, summed over all inputs, with input states sampled uniformly from $\{0,1\}^n$~\cite{shmulevich2004activities}. To examine the relationship between parameter-uniform sampling weight and these function-level properties, we considered all layer structures of NCFs with $n\in\{3,\ldots,6\}$ inputs (figure~\ref{fig:impact_detailed}). All quantities were computed exactly using BoolForge.

\subsection{Boolean network dynamics}

A Boolean network consists of $N$ nodes with binary states $x_i\in\{0,1\}$, each governed by a Boolean function $f_i$ that determines its state from those of its regulators. Under synchronous update, all nodes are updated simultaneously, defining a deterministic map $F:\{0,1\}^N\to\{0,1\}^N$. Because the state space is finite, every trajectory eventually reaches a periodic orbit, referred to as an \emph{attractor}. In biological Boolean network models, these attractors typically correspond to phenotypes or cell types. The set of states whose trajectories eventually reach a given attractor constitutes its \emph{basin of attraction}.

We characterized network dynamics using the number and mean length of attractors, coherence, and basin entropy. Coherence, also known as phenotypical robustness, describes the probability that two random binary network states $x,y \in \{0,1\}^N$ that differ in exactly one bit will transition to the same attractor~\cite{willadsenwiles,kadelka2013stabilizing,park2023models,kadelka2023modularity}. Basin entropy quantifies the diversity of attractor basin sizes and is defined as
\[
H_{\mathrm{basin}}=-\sum_{\alpha}p_\alpha\log p_\alpha,
\]
where $p_\alpha$ denotes the fraction of network states belonging to the basin of attractor $\alpha$~\cite{krawitz2007basin}.

\subsection{Generation and analysis of random Boolean networks}

Random Boolean networks consisted of $N=12$ nodes with constant in-degree $n\in\{2,\ldots,6\}$. For each node, $n$ regulators were selected uniformly at random without replacement from the remaining $N-1$ nodes, thereby excluding self-regulation. For each combination of in-degree and sampling measure, $10\,000$ independent networks were generated by assigning nested canalizing Boolean functions independently to all nodes.

Using BoolForge~\cite{kadelka2026boolforge}, network dynamics were analyzed under synchronous update using the exact state transition graph. We computed the number of attractors, mean attractor lengths, coherence, and basin entropy for every network and report in figure~\ref{fig:dynamics} the corresponding sample means together with 95\% confidence intervals.

\subsection{Meta-analysis of published Boolean network models}

To assess the biological implications of the sampling measure, we revisited the collection of 122 published Boolean gene regulatory network models compiled in~\cite{kadelka2024meta}. NCFs with three, four, or five inputs were grouped according to their canalizing layer structure. For each layer structure, fold enrichment was calculated as the observed frequency divided by the expected frequency under either parameter-uniform or function-uniform sampling. Average sensitivities of the corresponding layer structures were computed exactly using BoolForge~\cite{kadelka2026boolforge}, and enrichment trends were summarized separately for each in-degree by linear regression of log-transformed fold enrichment $E$ against average sensitivity $s$,
\[
\log E = m s + c.
\]

\subsection*{Data accessibility}
All code required to rerun the analyses and recreate all figures is available on Github at \url{https://github.com/ckadelka/boolforge-uniform-samplers}. The function-uniform sampling algorithms developed in this work are implemented in BoolForge v1.0.0~\cite{kadelka2026boolforge}, and exposed through the user-facing \texttt{random\_function} and \texttt{random\_k\_canalizing\_function} routines: \url{https://pypi.org/project/boolforge}

\subsection*{Author contributions}
A.G.: formal analysis, methodology, writing - original draft, review \& editing; C.K.: conceptualization, formal analysis, funding acquisition, methodology, software, writing - original draft, review \& editing

\subsection*{Funding}
This work was supported by the National Science Foundation (grant nos. DMS-2424632 (to C.K.)and DMS-2451973 (to C.K.)).

\appendix
\setlength{\parindent}{0pt}

\section*{Appendix: Mathematical details and proofs}

\subsection{Combinatorial multiplicity of layer structures}
\label{comb mul}

For fixed numbers of inputs $n$ and canalizing variables $k$, let $\#\mathrm{CF}_{n,\geq k}(\lambda)$ denote the number of distinct Boolean functions of canalizing depth at least $k$ whose first $k$ canalizing variables have layer structure $\lambda=(k_1,\dots,k_r)$, where $k_1+\cdots+k_r=k$. This number can be determined by counting the independent choices in their construction \cite{dimitrova2022revealing,he2016stratification}:
\begin{itemize}
\item $\binom{n}{k}$ ways to choose the $k$ canalizing variables;
\item $\frac{k!}{k_1!\cdots k_r!}$ ways to assign these variables to the $r$ ordered layers of sizes $k_1,\dots,k_r$, since permutations of variables within the same layer do not change the represented Boolean function~\cite{li2013boolean};
\item $2^k$ choices of canalizing input values, one for each of the $k$ canalizing variables;
\item $2$ choices for the first canalized output value $b_1$, since all subsequent canalized output values $b_2,\ldots,b_k$ are determined by the layer structure $\lambda$~\cite{kadelka2017influence}; and
\item $2^{2^{n-k}}-1$ choices for the core function on the remaining $n-k$ variables, which may be any Boolean function except the constant function equal to $b_k$. Since the canalizing depth is required only to be at least $k$, the core function may itself be canalizing.
\end{itemize}
By uniqueness of the canonical layer decomposition, distinct choices counted above correspond to distinct Boolean functions. Multiplying these factors gives
\[
\#\mathrm{CF}_{n,\geq k}(\lambda)
= \binom{n}{k}\,2^{k+1}\,
\frac{k!}{\prod_{i=1}^r k_i!}
\left(2^{2^{n-k}}-1\right).
\]
For fixed $n$ and $k$, only the factor $\prod_{i=1}^r k_i!$ depends on the layer structure $\lambda$. Therefore,
\[
\#\mathrm{CF}_{n,\geq k}(\lambda)
\propto
\frac{1}{\prod_{i=1}^r k_i!}
=: \omega(\lambda).
\]

Similarly, let $\#\mathrm{CF}_{n,=k}(\lambda)$ denote the number of distinct Boolean functions of exact canalizing depth $k$ with layer structure $\lambda=(k_1,\dots,k_r)$. The counting argument is identical except for the admissible core functions. To ensure that $k=k_1+\cdots+k_r$ is the exact canalizing depth, the core function must be non-canalizing. The number of admissible non-canalizing core functions is known~\cite[Theorem 5.4]{he2016stratification} and depends on $n$ and $k$, but not on the particular layer structure $\lambda$. Consequently,
\[
\#\mathrm{CF}_{n,=k}(\lambda)
\propto
\frac{1}{\prod_{i=1}^r k_i!}
= \omega(\lambda).
\]
Thus, for both prescribed minimal and prescribed exact canalizing depth, the number of distinct Boolean functions associated with a layer structure is proportional to $\omega(\lambda)$. In particular, layer structures consisting of many small layers correspond to more distinct Boolean functions than structures consisting of few large layers, because permutations of variables within a layer leave the Boolean function unchanged.

\subsection{Function-uniform sampling of canalizing layer structures}
\label{weight function}

As shown in Appendix~\ref{comb mul}, for fixed $n$ and $k$, the number of distinct Boolean functions associated with a canalizing layer structure
\[
\lambda=(k_1,\ldots,k_r), \qquad k_1+\cdots+k_r=k,
\]
is proportional to
\[
\omega(\lambda)
=
\frac{1}{\prod_{i=1}^r k_i!}.
\]
Consequently, uniform sampling over distinct Boolean functions requires sampling layer structures with probabilities proportional to $\omega(\lambda)$.

The layer structure of the $k$ canalizing variables is completely determined by their canalized output values $\mathbf b=(b_1,\ldots,b_k)$~\cite{kadelka2017influence}: consecutive canalizing variables belong to the same layer if and only if their canalized output values agree. Thus, after choosing $b_1$ uniformly from $\{0,1\}$, the layer structure can be constructed sequentially by deciding at each step whether
\[
b_i=b_{i-1},
\]
which extends the current layer, or
\[
b_i\neq b_{i-1},
\]
which closes the current layer and starts a new one. The probabilities of these decisions cannot in general be chosen uniformly, because different decisions leave different total combinatorial weights among the possible layer-structure completions. We therefore compute these remaining weights by dynamic programming (DP).

\begin{example}
\label{probability tree}
Consider NCFs with $n=4$ inputs. The possible layer structures and their relative weights are
\[
\begin{array}{c|cccc}
\lambda &(4)&(1,3)&(2,2)&(1,1,2)\\
\hline
24\,\omega(\lambda)&1&4&6&12.
\end{array}
\]
Hence their function-uniform probabilities are respectively
\[
\frac{1}{23},\qquad
\frac{4}{23},\qquad
\frac{6}{23},\qquad
\frac{12}{23}.
\]
After $b_1$ has been chosen, the decision $b_2=b_1$ leaves the structures $(4)$ and $(2,2)$, whose combined weight is $7/24$, whereas $b_2\neq b_1$ leaves $(1,3)$ and $(1,1,2)$, whose combined weight is $16/24$. Thus,
\[
\mathbb P(b_2=b_1)=\frac{7}{23},
\qquad
\mathbb P(b_2\neq b_1)=\frac{16}{23}.
\]
Conditioning recursively on the remaining structures gives the probability tree below.

\begin{center}
\begin{tikzpicture}[
  level distance=1.5cm,
  level 1/.style={sibling distance=4cm},
  level 2/.style={sibling distance=2cm},
  every node/.style={inner sep=2pt, font=\small},
  edge label/.style={draw=none, midway, font=\small, fill=white}
]
\node {$b_1$}
  child {
    node {$b_2=b_1$}
    child {
      node {$b_3=b_2$}
      child {
        node [label=below:{$\lambda=(4)$}] {$b_4=b_3$}
        edge from parent node[edge label,left] {$1$}
      }
      edge from parent node[edge label,left] {$1/7$}
    }
    child {
      node {$b_3\neq b_2$}
      child {
        node [label=below:{$\lambda=(2,2)$}] {$b_4=b_3$}
        edge from parent node[edge label,left] {$1$}
      }
      edge from parent node[edge label,right] {$6/7$}
    }
    edge from parent node[edge label,left] {$7/23$}
  }
  child {
    node {$b_2\neq b_1$}
    child {
      node {$b_3=b_2$}
      child {
        node [label=below:{$\lambda=(1,3)$}] {$b_4=b_3$}
        edge from parent node[edge label,left] {$1$}
      }
      edge from parent node[edge label,left] {$1/4$}
    }
    child {
      node {$b_3\neq b_2$}
      child {
        node [label=below:{$\lambda=(1,1,2)$}] {$b_4=b_3$}
        edge from parent node[edge label,left] {$1$}
      }
      edge from parent node[edge label,right] {$3/4$}
    }
    edge from parent node[edge label,right] {$16/23$}
  };
\end{tikzpicture}
\end{center}

Multiplying the conditional probabilities along each path gives
\[
\begin{array}{c|cccc}
\lambda &(4)&(1,3)&(2,2)&(1,1,2)\\
\hline
\mathbb P(\lambda)
&\frac{1}{23}&\frac{4}{23}&\frac{6}{23}&\frac{12}{23},
\end{array}
\]
as required.
\end{example}

\subsection{Dynamic programming recursion}
\label{DP Recursion}

After $b_1$ has been chosen uniformly at random, the sequential construction of the canalized outputs $\mathbf b$ is described by a state $(m,s)$, where $m$ is the number of canalized outputs that remain to be assigned and $s$ is the size of the current open canalizing layer. The initial state is therefore
\[
(m,s)=(k-1,1).
\]

From a state $(m,s)$ with $m>0$, there are two possible transitions:
\begin{enumerate}
    \item Extending the current layer by setting $b_i=b_{i-1}$ gives
\[
(m,s)\longrightarrow(m-1,s+1).
\]
\item Closing the current layer by setting $b_i\neq b_{i-1}$ contributes the factor $1/s!$ to the weight $\omega(\lambda)$, starts a new layer of size one, and gives
\[
(m,s)\longrightarrow(m-1,1).
\]
\end{enumerate}
Let $W(m,s)$ denote the total combinatorial weight of all valid completions from state $(m,s)$, including the eventual contribution of the currently open layer. The two possible transitions give the recursion
\[
W(m,s)
=
W(m-1,s+1)
+
\frac{1}{s!}W(m-1,1),
\qquad m\geq 1.
\]
The boundary conditions depend on the class of functions being sampled. For functions of exact canalizing depth $k<n$, a final canalizing layer of size one is permitted because it is followed by a non-canalizing core function. Hence
\[
W(0,s)=\frac{1}{s!},
\qquad s\geq 1.
\]
For NCFs, where $k=n$, the final layer must contain at least two variables in the unique layer decomposition~\cite{li2013boolean}. Therefore,
\[
W(0,s)=
\begin{cases}
0, & s=1,\\[2mm]
\dfrac{1}{s!}, & s\geq 2.
\end{cases}
\]
In other words, when sampling an NCF, a construction that would terminate with a singleton final layer is inadmissible.

The DP weights directly determine the transition probabilities. From any state $(m,s)$ with $m>0$ and $W(m,s)>0$, the algorithm chooses
\[
\mathbb P(\mathrm{extend}\mid m,s)
=
\frac{W(m-1,s+1)}{W(m,s)}
\]
and
\[
\mathbb P(\mathrm{close}\mid m,s)
=
\frac{\frac{1}{s!}W(m-1,1)}{W(m,s)}.
\]
These probabilities sum to one by the recursion for $W(m,s)$.

For example, in the $n=k=4$ case of Example~\ref{probability tree},
\[
W(3,1)
=
\frac{1}{4!}
+\frac{1}{1!3!}
+\frac{1}{2!2!}
+\frac{1}{1!1!2!}
=
\frac{23}{24},
\]
whereas extending the first layer leads to
\[
W(2,2)
=
\frac{1}{4!}
+\frac{1}{2!2!}
=
\frac{7}{24}.
\]
Consequently,
\[
\mathbb P(b_2=b_1)
=
\frac{W(2,2)}{W(3,1)}
=
\frac{7}{23},
\]
in agreement with Example~\ref{probability tree}.

Since there are $O(k^2)$ possible states $(m,s)$ and each DP entry requires constant time once entries with smaller $m$ have been computed, the table of weights can be constructed in $O(k^2)$ time. Once the table has been computed, generating a layer structure requires exactly $k-1$ sequential decisions and therefore takes $O(k)$ time. The same DP table can be reused for repeated samples with the same value of $k$ and the same boundary condition.

\subsection{Correctness of the dynamic programming sampler}
\label{Correctness}

\begin{theorem}[Correctness of the DP sampling algorithm]
\label{thm:correctness}
Let $\lambda=(k_1,\dots,k_r)$ be a valid layer structure for $k$ canalizing variables, with
\[
\omega(\lambda)
=
\frac{1}{\prod_{i=1}^r k_i!}.
\]
Starting from $(m,s)=(k-1,1)$, let the DP algorithm use the transition probabilities defined above. Then
\[
\mathbb P(\lambda)
=
\frac{\omega(\lambda)}{W(k-1,1)}.
\]
Thus, the algorithm samples layer structures with probability proportional to the number of distinct Boolean functions having that layer structure.
\end{theorem}

\begin{proof}
Consider any state $(m,s)$ and any valid completion $\lambda$ reachable from that state. Let $\omega_{m,s}(\lambda)$ denote the weight of a particular completion $\lambda$ from state $(m,s)$, defined as the product of the reciprocal factorials of the sizes of all layers completed from that state onward, including the current open layer once it is eventually completed. By definition,
\[
W(m,s)
=
\sum_{\lambda\in\Lambda_{m,s}}
\omega_{m,s}(\lambda),
\]
where the sum ranges over all valid completions from $(m,s)$.

We show by induction on $m$ that the probability of generating a particular valid completion $\lambda$ from state $(m,s)$ is
\[
P_{m,s}(\lambda)
=
\frac{\omega_{m,s}(\lambda)}{W(m,s)}.
\]

For $m=0$, no canalized outputs remain to be assigned. If the state is admissible, the current layer of size $s$ is the only remaining layer and
\[
\omega_{0,s}(\lambda)=W(0,s)=\frac{1}{s!},
\]
so $P_{0,s}(\lambda)=1$. In the inadmissible NCF state $(0,1)$, there is no valid completion.

Now suppose the statement holds for all states with fewer than $m$ remaining outputs. If $\lambda$ is obtained by extending the current layer, the first transition is
\[
(m,s)\longrightarrow(m-1,s+1)
\]
with probability
\[
\frac{W(m-1,s+1)}{W(m,s)}.
\]
No layer is closed at this transition, so
\[
\omega_{m,s}(\lambda)
=
\omega_{m-1,s+1}(\lambda).
\]
By the inductive hypothesis,
\[
P_{m,s}(\lambda)
=
\frac{W(m-1,s+1)}{W(m,s)}
\frac{\omega_{m-1,s+1}(\lambda)}
     {W(m-1,s+1)}
=
\frac{\omega_{m,s}(\lambda)}{W(m,s)}.
\]

If instead $\lambda$ is obtained by closing the current layer, the first transition is
\[
(m,s)\longrightarrow(m-1,1)
\]
with probability
\[
\frac{\frac{1}{s!}W(m-1,1)}{W(m,s)}.
\]
Let $\lambda'$ denote the remaining completion after the current layer has been closed. Then
\[
\omega_{m,s}(\lambda)
=
\frac{1}{s!}\omega_{m-1,1}(\lambda').
\]
Using the inductive hypothesis,
\[
P_{m,s}(\lambda)
=
\frac{\frac{1}{s!}W(m-1,1)}{W(m,s)}
\frac{\omega_{m-1,1}(\lambda')}
     {W(m-1,1)}
=
\frac{\omega_{m,s}(\lambda)}{W(m,s)}.
\]
Thus the result holds in both cases and hence for all $m$.

At the initial state $(k-1,1)$, the completion weight is precisely
\[
\omega_{k-1,1}(\lambda)=\omega(\lambda),
\]
which gives
\[
\mathbb P(\lambda)
=
\frac{\omega(\lambda)}{W(k-1,1)}.
\]
\end{proof}

By Appendix~\ref{comb mul}, for fixed $n$ and $k$ the number of distinct Boolean functions associated with a layer structure $\lambda$ is proportional to
\[
\omega(\lambda)
=
\frac{1}{\prod_{i=1}^r k_i!}.
\]
Theorem~\ref{thm:correctness} therefore shows that the DP algorithm samples each layer structure with probability proportional to the number of distinct Boolean functions having that structure. Conditional on the sampled layer structure, the canalizing variables, their assignment to layers, canalizing input values, the initial canalized output, and the admissible core function are sampled uniformly from their respective choices. Consequently, every distinct Boolean function of the prescribed class is generated with equal probability.

\subsection{Sampling $k$-canalizing functions uniformly at random}
\label{random sampling}

The DP algorithm yields a canalizing layer structure
\[
\lambda=(k_1,\ldots,k_r),
\qquad
k_1+\cdots+k_r=k
\]
for the $k$ canalizing variables.

When $k<n$, the remaining $n-k$ variables define a core function that is evaluated only when all canalizing variables receive their respective non-canalizing input values. If the goal is to sample uniformly from all Boolean functions with \emph{exact} canalizing depth $k$, we can therefore sample the core function uniformly at random from the set of all non-canalizing functions in $n-k$ variables, using rejection sampling. In particular, functions with exact canalizing depth $k=n-1$ do not exist, since a Boolean function with a single-variable core necessarily has canalizing depth $n$. This procedure is efficient because the probability that a random Boolean function is canalizing decreases exponentially with the number of inputs~\cite{he2016stratification,kadelka2024meta}.

However, if the goal is to sample uniformly from all Boolean functions with canalizing depth \emph{at least} $k$ (that is, from all $k$-canalizing functions), then the core function itself may be canalizing. In this case, the first canalizing layer of the core function may merge with the final layer of $\lambda$ whenever both layers have the same canalized output value. Such a merge increases the effective size of the final layer and creates a combinatorial overcount.

\begin{theorem}[Rejection correction for merged layers]
\label{thm:merged}
Let $k_r$ be the size of the last canalizing layer in $\lambda$, and let $k_{\mathrm{core}}$ be the size of the first canalizing layer of the core function, provided that its canalized output value agrees with the canalized output value of the final layer of $\lambda$. Then the merged final layer has size
\[
k' = k_r + k_{\mathrm{core}}.
\]
Without correction, any function arising from such a merge is oversampled by the factor
\[
\binom{k'}{k_r},
\]
because any subset of $k_r$ variables chosen from the $k'$ merged variables could have been designated as the final canalizing layer of the sampled representation.

To restore uniform sampling, such a function must therefore be accepted with probability
\[
\frac{1}{\binom{k_r+k_{\mathrm{core}}}{k_r}}.
\]
\end{theorem}

\begin{proof}
Suppose the uncorrected procedure produces a function whose final canalizing layer of size $k_r$ merges with the first canalizing layer of the core function of size $k_{\mathrm{core}}$. The resulting effective layer has size
\[
k' = k_r + k_{\mathrm{core}}.
\]
There are exactly
\[
\binom{k'}{k_r}
\]
ways to choose which $k_r$ of the $k'$ merged variables are treated as belonging to the sampled final canalizing layer, with the remaining $k_{\mathrm{core}}$ variables treated as part of the core. Each such choice yields the same Boolean function, so the uncorrected algorithm generates that function $\binom{k'}{k_r}$ times.

Accepting the proposed function with probability
\[
\frac{1}{\binom{k'}{k_r}}
\]
exactly compensates for this multiplicity. Hence every Boolean function is accepted with the same overall probability.
\end{proof}

When a rejection occurs, both the canalizing layer structure $\lambda$ and the core function are discarded, and the entire procedure is restarted. This is necessary because the DP algorithm samples layer structures with probability
\[
P(\lambda)\propto \omega(\lambda),
\]
which is the correct weighting for function-uniform sampling. If $\lambda$ were retained after a rejection and only the core function were resampled, then the same layer structure would be reused across multiple proposals and would therefore be sampled more often than prescribed by $\omega(\lambda)$. Restarting the full procedure ensures that each $(\lambda,\mathrm{core})$ pair is proposed with the correct probability, while the rejection step corrects only for the overcounting caused by merged layers.

The rejection step remains efficient even in the worst case. The probability of rejection is at most $5/8$, and this maximum occurs for the smallest admissible core size of two, i.e., when $k=n-2$. Moreover, the probability that a random function is canalizing decreases exponentially fast as the number of variables increases~\cite{he2016stratification}; since a merge can occur only when the randomly sampled core function is itself canalizing, the rejection probability is bounded above by an exponentially decreasing function of $n-k$. Consequently, the expected number of proposals required to obtain an accepted sample is at most
\[
\frac{1}{1-5/8}=\frac{8}{3},
\]
and approaches $1$ as the number of core variables increases. Thus, the rejection correction does not alter the practical scalability of the sampling procedure.

\bibliography{refs}

\end{document}